\documentclass{article}

\usepackage[latin1]{inputenc}
\usepackage{subfigure}
\usepackage{fullpage}
\usepackage{times}
\usepackage{color}
\usepackage{ulem} 
\usepackage{amssymb}
\usepackage[cmex10]{amsmath}
\usepackage{amsthm}
\usepackage{xspace}
\usepackage[nocompress]{cite}
\usepackage[noend, ruled]{algorithm2e}
\usepackage{url}
\usepackage{hyperref}

\newtheorem{theorem} {Theorem}
\newtheorem{lemma}[theorem] {Lemma}
\newtheorem{claim}[theorem] {Claim}

\newcommand{\CH}{{\ensuremath{{C}}}\xspace}
\newcommand{\VRF}{{\ensuremath{{V}}}\xspace}
\newcommand{\Cp}{{\ensuremath{\mathit{WP_{\CH}}}}\xspace}
\newcommand{\Cw}{{\ensuremath{\mathit{MP_{W}}}}\xspace}
\newcommand{\Ct}{{\ensuremath{\mathit{WC_{T}}}}\xspace}
\newcommand{\Cv}{{\ensuremath{\mathit{MC_{\VRF}}}}\xspace}
\newcommand{\SM}{{\ensuremath{\mathit{MC_{A}}}}\xspace}
\newcommand{\SW}{{\ensuremath{\mathit{WB_{A}}}}\xspace}
\newcommand{\Bc}{{\ensuremath{\mathit{MB_{R}}}}\xspace}
\newcommand{\gameone}{{\ensuremath{1\!\!:\!\!1}}\xspace}
\newcommand{\gametwo}{{\ensuremath{1\!\!:\!\!1^n}}\xspace}
\newcommand{\gametwos}{{\ensuremath{1:1^n}}\xspace}
\newcommand{\gamethree}{{\ensuremath{0\!\!:\!\!n}}\xspace}
\newcommand{\gamethrees}{{\ensuremath{0:n}}\xspace}
\newcommand{\gamefour}{{\ensuremath{1\!\!:\!\!n}}\xspace}

\newcommand{\modelone}{{\ensuremath{{R}_{\rm m}}}\xspace}
\newcommand{\modeltwo}{{\ensuremath{{R}_{\rm a}}}\xspace}
\newcommand{\modelthree}{{\ensuremath{{R}_{\emptyset}}}\xspace}
\newcommand{\parhead}[1]{\noindent{\textbf{#1.}\xspace}}
\newcommand{\cgaa}[1]{#1}
\newcommand{\mmaa}[1]{#1}
\newcommand{\cga}[1]{#1}   

\newcommand{\mma}[1]{#1}   
\newcommand{\afa}[1]{#1} 

\date{}

\title{Algorithmic Mechanisms for Reliable Internet-based Computing under Collusion
\thanks{This work is supported in part by the Cyprus Research Promotion Foundation
grant T$\Pi$E/$\Pi\Lambda$HPO/0609(BE)/05, Comunidad de Madrid grant S2009TIC-1692, 
Spanish MICINN grant TIN2008--06735-C02-01,
and the National Science Foundation (CCF 1114930, CCF 0937829).
A preliminary version of this work appears in the Proceedings of NCA 2008, 
pages 315--324.}}

\author{
Antonio~Fern\'andez~Anta\thanks{
Institute IMDEA Networks, Madrid, Spain.
Email: \href{mailto:antonio.fernandez@imdea.org}{antonio.fernandez@imdea.org}.} 
\and
Chryssis~Georgiou\thanks{
Dept. of Computer Science, Univ. of Cyprus, Cyprus.
Email: \href{mailto:chryssis@cs.ucy.ac.cy}{chryssis@cs.ucy.ac.cy}.} 
\and
Miguel~A.~Mosteiro\thanks{
Dept. of Computer Science, Kean University, USA \& GSyC, Univ. Rey Juan Carlos, Spain.
Email: \href{mailto:mmosteir@kean.edu}{mmosteir@kean.edu}.}
}

\begin{document}

\maketitle

\begin{abstract}
In this work, using a game-theoretic approach, cost-sensitive mechanisms that lead to reliable
Internet-based computing are designed.
In particular, we consider Internet-based master-worker computations, where a
master processor assigns, across the Internet, a computational task to a
set of potentially untrusted worker processors and collects their responses.
Workers may collude in order to increase their benefit.

Several game-theoretic models that
capture the nature of the problem are analyzed, and algorithmic mechanisms that, for
each given set of cost and system parameters, achieve high
reliability are designed. 
Additionally, two specific realistic system scenarios are studied.
These scenarios are a system of volunteer computing like SETI, and  a company that buys computing cycles from Internet
computers and sells them to its customers in the form of a task-computation service.
Notably, under certain conditions, non redundant allocation yields the best trade-off between cost and reliability.\bigskip

\end{abstract}

\section{Introduction}
\label{sec:Intro}
%
\paragraph{Motivation.}
As traditional one-processor machines have limited computational
resources, and powerful parallel machines are very expensive to obtain and
maintain, the Internet is emerging as 
a viable computational platform 
for processing complex computational jobs. Several
Internet-oriented systems and protocols have been designed to operate
on top of this global computation infrastructure; examples include
Grid systems~\cite{EGEE,TeraGrid}, 
the ``@home'' projects~\cite{boinc}, such as SETI~\cite{SETI} (a classical example of \emph{volunteer computing}),  
and peer-to-peer computing--P2PC~\cite{P2PC,CCS}. 
Although the potential is great, the use of Internet-based computing 
is limited by the untrustworthiness nature of the
platform's components~\cite{UDC,boinc}.
Let us take SETI as an example. In SETI, data is distributed for
processing to millions of voluntary machines around the world. At a
conceptual level, in SETI there is a machine, call it the
\emph{master}, that sends jobs, across the Internet, to these computers,
call them the \emph{workers}. These workers execute and report back the result
of the task computation. However, these workers are not trustworthy,
and hence might report incorrect results. In SETI, the master attempts to
minimize the impact of these bogus results by assigning the same task
to several workers and comparing their outcomes (that is, {\em
redundant} task allocation is employed~\cite{boinc}), 
but there are also other methods~\cite{KCWB02,Duetal04,W06}.

In this paper, Internet-based master-worker computations \mma{are studied} from a game-theoretic point
of view. Specifically, these computations \mma{are modeled} as games where each worker chooses
whether to be {\em honest} (that is, compute and return the correct task result) or a
{\em cheater} (that is, fabricate a bogus result and return it to the master).
\mma{Additionally,} cost-sensitive mechanisms (algorithms) 
that provide the necessary incentive for the workers to truthfully compute
and report the correct result \mma{are designed}. The objective is to maximize the
probability of the master of obtaining the correct result while
minimizing its cost (or alternatively, increasing its benefit). 
\mma{In particular,} we identify and propose mechanisms for two paradigmatic applications.
Namely, a computing system as the aforementioned SETI where computing processors volunteer to donate part of their processing time,
and a second scenario where a company distributes computing tasks among contractor processors that get an economic reward in exchange.

Although the presentation is self-contained, it is assumed familiarity with basic concepts in Game Theory. For further details please refer to~\cite{osborne}.

\paragraph{Background and Prior/Related Work.}
Prior examples of Game Theory
in distributed computing include work on Internet
routing~\cite{KP99, RT02, MS07}, resource/facility location and
sharing~\textcolor{black}{\cite{Hall04, Fotakis06, HR_STOC08}}, containment of viruses
spreading~\cite{Roger06}, secret sharing~\cite{Halp04, Halp06}, and
task computations~\cite{CCS}. For
more discussion on the connection between Game Theory and computing we
refer the reader to the survey by Halpern~\cite{Survey07} and the book
by Nisan et al.~\cite{Nisan07}.

In traditional distributed computing, the behavior of the system
components (i.e., processors) is characterized a priori as either good
or bad, depending on whether they follow the prescribed protocol or
not.  In Game Theory, processors are assumed to act on their
own {\em self-interest} and they do not have an a priori established behavior.
Such processors are usually referred as {\em
rational}~\cite{UDC,Halp06}.  
In other words, the processors decide on how to act in
an attempt to increase their own benefit, or alternatively to lower
their own cost.  

In {\em Algorithmic Mechanism Design}~\cite{Nisan01, papa05, Halp06, SIGACT18}, 
games are designed to provide the necessary
incentives so that processors' interests are best served by
acting ``correctly.'' The usual practice is to provide some reward
(resp. penalty) should the processors (resp. do not) behave
as desired. 

\textcolor{black}{In~\cite{S_FGCS,SRDS, ALEX}} reliable master-worker computations have been
considered by redundant task-allocation.
In these works probabilistic guarantees of obtaining the correct result while minimizing the cost (number of workers chosen to perform the task \textcolor{black}{or amount of redundant allocation}) are also shown.
However, a traditional distributed computing approach is used, in which the behavior of each worker is predefined.
In this paper, much richer payoff parameters are studied and the behavior of each worker is not predefined, introducing new challenges that naturally drive to a game-theoretic approach.

Two other related works~\cite{CSAMSDAgridCollusion,STgridCollusion} where the worker behavior is predefined consider collusion in desktop grid computing. In both proposals, the goal is to identify colluders by means of a statistical analysis that requires the processors to compute multiple times. In the present paper, we study the more challenging problem of dealing with collusion when each processor computes only once.
 
\textcolor{black}{Previous work directly related to the present paper is included in~\cite{CCS} and~\cite{FGM:mechdesign},
where master-worker computations in a game-theoretic model are also studied.}
In~\cite{CCS}, the master can audit the results returned by rational workers with a tunable probability. Bounds for that audit probability are computed to guarantee that workers have incentives to be honest in three scenarios: redundant allocation with and without collusion\footnote{Cooperation among various workers concealed from the master.}, and single-worker allocation. They conclude that, in their model, single-worker allocation is a cost-effective mechanism specially in presence of collusion. 
In the present paper, a general study of how to carry out the computation depending on the system parameters (reward model, payoffs, bounds on utility or probability of incorrect result, etc.) is given (see Tables~\ref{table:seti} and~\ref{table:contractor}). 
For some cases studied here, the conclusion that redundant allocation does not help is also extracted. 
(E.g., in a scenario where only the number of workers is a choice and the result must be correct with probability $1$ for any payoff values.)
However, this may not be always the case.
For scenarios where the system parameters yield redundant allocation as the best approach, the analysis provided may be used by the master to choose games and reward models conveniently. Additionally, our work complements that work in various ways, such as studying more games,  including a richer payoff model, or considering probabilistic cheating. Finally, useful trade-offs between the benefit of the master and the probability of accepting an incorrect result are shown for the one-round protocol we propose.
\textcolor{black}{After the conference version of the present paper, in~\cite{FGM:mechdesign} and~\cite{CFGM:NCA11}, we extended the potential worker personalities with malicious and altruistic behaviors, but without considering collusion. Under this model, we explore only a subset of the games studied here.}

Distributed computation in presence of selfishness was also studied within the scope of \emph{Combinatorial Agencies} in Economics~\textcolor{black}{\cite{BFN06,Babaioff06,ES09,BFN_SAGT09}. 
The basic model considered is a combinatorial variant of the classical 
principal-agent problem~\cite{MWG95}: A master (principal) must
motivate a collection of workers (agents) to exert costly effort on
the master's behalf, but the workers' actions are hidden from the
master. Instead of focusing on each worker actions, the focus is
on complex combinations of the efforts of the workers that influence
the outcome. In~\cite{BFN06}, where the problem was first introduced,
the goal was to study how the utility of the master is affected if the equilibria 
space is limited to \emph{pure} strategies. I.e., equilibria computation where it is assumed that the players do not randomize their choice and, instead, deterministically choose among one of the available strategies. To that extent, the computation of a few Boolean functions is evaluated. In~\cite{Babaioff06} mixed strategies were considered: if the parameters of the problem yield multiple mixed equilibrium points, it is assumed that workers accept one ``suggested'' by the master. This is contrasted with our work as we require the master
to enforce a single equilibrium point (referred as {\em strong implementation} in~\cite{BFN06}). 
The work in~\cite{ES09} investigates the effect of auditing by allowing
the master to audit some workers (by random sampling) and verify their work. In our work,
the master decides probabilistically whether to verify all workers or none.  
In general, the spirit of the framework considered in Combinatorial agency is similar to the one we consider in the present work in the
sense that there is a master wishing a specific outcome and it must provide necessary
incentives to rational workers so to reach that outcome (exerting effort can be
considered as the worker performing the task, and not, as the worker not performing
the task and reporting a bogus result). However, there are several
differences. The main difference is that in our framework, the worker actions cannot
really be viewed as {\em hidden}. The master receives a response by each worker and it is
aware that either the worker has truthfully performed the task or not. The outcome
is affected by each worker's action in the case that no verification is performed 
(in a similar fashion as the majority boolean ``technology'' in Combinatorial agency) 
but via verification the master can determine the exact strategy used by each worker
and apply a specific reward/punishment scheme. In the framework considered
in combinatorial agency, the master witnesses the outcome of the computation, but it has
no knowledge of the possible actions that the worker might take. 
For this purpose, the master needs to devise contracts
for each worker based on the observed outcome of the computation and not on each
worker's possible action (as in our framework). Another important difference
includes the fact that our scheme considers worker punishment, as opposed to
the schemes in combinatorial agency where workers cannot be fined (limited liability
constraint); this is possible in our framework as workers' actions are ``contractible''
(either it performs a task or not).     
}

\textcolor{black}{Monderer and Tennenholtz~\cite{MT_EC03} consider a master-worker
framework where the master wishes to influence the behavior of rational workers
in a game, which is not under the master's control. Namely, the master cannot design
a new game, cannot enforce worker's behavior, cannot penalize the workers and cannot
prohibit strategies available to the workers. In the framework we consider in the present
work, the master might not be able to control all the parameters of the game (e.g., the
various reward schemes and other system parameters that are part of the game) but
it can indirectly influence the behavior of the workers by the one parameter that
it surely controls: the probability of auditing the results returned by the workers.  
}

A somewhat related work is~\cite{DM:lottery} in which they face the problem of bootstrapping a P2P computing system, in the presence of rational peers. The goal is to incentivize peers to join the system, for which they propose a scheme that mixes lottery psychology and multilevel marketing. In our setting, the master could use their scheme to recruit workers. We assume in this paper that enough workers are willing to participate in the computation.

\textcolor{black}{Du et al.~\cite{Duetal04} present a commitment-based sampling
scheme for cheater detection in Grid computing that is based on Merkle trees. Their
model considers a task as a domain of inputs $D=\{x_1,x_2,\ldots,x_n\}$ and a
function $f$ such that each $f(x)$ must be computed for all $x\in D$. Instead of 
using redundant task allocation (as in our work), or double-checking the 
worker's computation for each $x$, the master allocates the task to a worker
and randomly selects a small number of inputs from $D$ and double-checks these
results. As the authors point out, their technique works well for input domains
of large size (large $|D|$) but does not for small $|D|$; our redundant allocation
scheme could be considered for such small-sized input domains.}

\textcolor{black}{Kuhn et al.~\cite{KSW_IPDPS08} consider a distributed verification mechanism for
computational Grids. Instead of having the master checking and detecting
cheaters, their mechanism passes this ``responsibility'' to the workers. The master
distributes two different kinds of tasks to workers, regular computation tasks, and
checking units. For the first type the worker is required to compute its result, as oppose
to checking tasks that require the worker to perform a number of checks
for different results reported by other workers. Workers are encouraged to act 
correctly via credit points (that can be either used in a form of prestige, or be converted to
real money). The work in the present paper focuses on having the master to
obtain the correct result within some probability of success rather than detecting 
cheaters (partially this is achieved, but it is not the main objective of the master).
For cheater detection to be beneficiary, the computation must be run over several rounds.
Our work, instead, considers a one-shot protocol that enables fast termination (the benefit of
one-round mechanisms is partially supported by the work of Kondo et al.~\cite{KCWB02} 
that have demonstrated experimentally that tasks may take much more than 
one day of CPU time to complete).  
}

\paragraph{Framework.}
We consider a distributed system consisting of a master processor that assigns a computational task to a set of workers to compute and return the task result. 
The tasks considered in this work are assumed to have a unique solution. Although such limitation reduces the scope of application of the mechanisms presented~\cite{TACB05}, there are plenty of computations where the correct solution is unique. E.g., any mathematical function. 
Notice that we consider one-shot computations only. I.e., in this work we do not consider multiple computations where workers may accumulate reputation according to past behavior.

It is assumed that the master has the possibility of verifying whether the value returned by a worker is the correct result of the task.
It is also assumed that verifying an answer is more efficient than computing the task~\cite{GS:secureDC} (e.g., $NP$-complete problems if $P\neq NP$), but \afa{the correct result of the computation is not obtained if the verification fails.} Therefore, by verifying, the master does not necessarily obtain the correct answer (e.g., when all workers cheat)%
\footnote{Alternatively, one might assume that the master verifies by simply
performing the task and checking the answers of the workers. Our analysis can
easily be modified to accommodate this different model.}.
As in~\cite{CCS, DM:lottery}, workers are assumed to be rational and seek to maximize their benefit, i.e., they are not destructively malicious. We note that
this assumption can conceptually be justified by the work of Shneidman and Parkes~\cite{rational}
where they reason on the connection of rational players--of Algorithmic 
Mechanism Design--and workers in realistic P2P systems.
Furthermore, we do not consider non-intentional errors produced by hardware or software problems.

The general protocol used by master and workers is the following.
The master process assigns the task to $n$
workers.  Each worker processor $i$ cheats with
probability $p_{\CH}^{(i)}$ and the master processor verifies the answers with
some probability $p_{\VRF}$.  If the master processor verifies, it rewards the honest workers and penalizes the cheaters. If the master 
does not verify, it accepts the answer returned by the majority of workers. However, it does not penalize any worker given that the majority can be actually cheating. Instead, the master rewards workers according to one of the three following models.
Either the master rewards the majority only (\emph{Reward Model} \modelone),
or the master rewards all workers independently of the returned value 
(\emph{Reward Model} \modeltwo), or the master does not reward at all 
(\emph{Reward Model} \modelthree). 

The model used in this paper comprises the following \cgaa{form of collusion} (that covers realistic types of collusions such as Sybil attacks\cite{D:sybil}). Workers form colluding groups. Within
the same group workers act homogeneously, i.e., either all choose to cheat, or all choose to be honest, perhaps randomizing their decision by tossing a unique coin. In the case that, within the group, all workers choose to be honest, then only one of them computes the task, and all of them return that result to the master (in this way they avoid the cost of all of them executing the task). In the case that all workers choose to cheat, then they simply agree on a bogus result and send that to the master. In addition, we assume that all ``cheating groups'' return the same incorrect answer.
Both assumptions (homogeneous behavior within groups and unique incorrect answer) are adversarial.
Since the master accepts the majority, this behavior maximizes the chances of cheating the master. 
Being this the worst case~\textcolor{black}{(see also~\cite{S_FGCS})}, it subsumes models where cheaters do not necessarily return the same answer. 
\cgaa{Note that \mmaa{this behavior} can be viewed \mmaa{also} as a form
of collusion.} 
However, this observation does not imply that cheaters coordinate among them such behavior.
We also assume that if a worker does not perform the task, then it is (almost) impossible to guess the correct answer (i.e., the probability is negligible). The master, of course, is not aware of the collusions.

Given the protocol above, the game is defined by a set of parameters that include rewards to the
workers that return the correct value \mma{and} punishments to the workers
that cheated (that is, returned the incorrect result and ``got caught''). 
Hence, the game is played between the master and the workers, where the first wants to
obtain the correct result with a desired probability, while obtaining a desired utility
value (in expectation), and the workers decide whether to be honest or cheaters, 
depending on their expected utility gain or loss.
In this paper, we design several games and study the conditions under
which unique Nash equilibria \mma{(NE)} are achieved.
The reason for uniqueness is to force all workers to the same strategy;
this is similar to {\em strong implementation} in Mechanism Design, cf.,~\cite{BFN06}. 
(Multiple equilibria could be considered, making further assumptions about the procedure that workers follow to choose one of them. Although the approach might be promising in terms of the utility for the master, in this work correctness is the priority, which as shown later the mechanisms presented here guarantee.)
Each NE results in a different benefit for the master and a different probability of accepting an incorrect result. Thus, the master can choose some game conditions so that a unique NE that best fits its goals is achieved.

\paragraph{Contributions.}
The main contributions of this paper are:\vspace{.3em}

\parhead{1}
The identification of a collection of realistic payoff parameters that
allow to model Internet-based master-worker computational environments in game
theoretic terms. These parameters can either be fixed because they are system
parameters or be chosen by the master.\vspace{.3em}

\parhead{2}
The definition of four different games that the master can force to be
played: (a) A game between the master and a single worker, (b) a game
between the master and a worker, played $n$ times (with different
workers), (c) a game with a master and $n$ workers, and (d) a game of
$n$ workers in which the master participates indirectly. 
\cgaa{Games (c) and (d) consider collusions, game (a) considers no collusions 
as there is only one worker,} 
\mmaa{and game (b) only considers singleton groups, where all cheaters return the same value.}
Together with the three reward models defined above, we have overall 
defined twelve games among which the master can choose the most convenient 
to use in each specific context.\vspace{.3em}

\parhead{3}
The analyses of all the games under general payoff models, and the
characterization of conditions under which a unique Nash Equilibrium
point is reached for each game and each payoff-model. These analyses lead to
mechanisms that the master can run to trade cost and reliability.\vspace{.3em}

\parhead{4}
The design of mechanisms for two specific realistic scenarios, 
to demonstrate the utility of the analysis. These scenarios reflect, in their
fundamental elements, (a) a system of volunteer computing like
SETI, and (b) a company that buys computing cycles from Internet
computers and sells them to its customers in the form of a
task-computation service. 
The analysis provided for these scenarios in Tables~\ref{table:seti} and~\ref{table:contractor} comprise implicitly a mechanism to decide how to carry out the computation. More specifically, depending on the various parameters of the problem, such as the instance of payoff values, the desired probability of obtaining the correct answer, or a (possibly negative) lower bound on the master's utility, the master may obtain from these tables the game to be played and the reward model that maximally benefit its goal: accuracy, utility, or both.\vspace{.3em}

\parhead{5}
As examples of the use of the mechanism designed, we consider computations where the result must be correct with probability $1$ for any instance of payoff values. Under such requirement, our results show that for scenario (a) the best choice is non-redundant allocation, even with only singleton colluding groups. Furthermore, in this case we show that to obtain always the correct answer it is enough to verify with arbitrarily small probability.
Regarding examples of scenario (b), under the same requirement, we evaluate the mechanism for settings where one of three parameter values can be chosen: the number of workers, the worker's punishment for being caught cheating, or the cost of computing the task.
If the master only chooses the number of workers $n$, we show that, again even with only singleton colluding groups the best choice is non-redundant allocation. However, in order to achieve correctness,
the required probability of verifying can now be large.
When only one of the other two parameters is a choice of the master, namely either the worker's punishment for being caught cheating or the cost of computing the task, we show that the best game is not unique, and it depends on the rest of parameters of the system.\vspace{.3em} 

In general our analysis depicts the tradeoffs between cost and reliability  for a wide range of system parameters,  
payoffs, and reward  models.

\paragraph{Paper Structure.}
In Section~\ref{sec:Defs} we provide basic definitions to be used
throughout the paper.  In Section~\ref{sec:Gs} we present and analyze
the games proposed. In Section~\ref{sec:MD} the
mechanisms for the two realistic scenarios are designed. 
Finally, Section~\ref{sec:Conc} presents conclusions and future lines of work.


\section{Definitions}
\label{sec:Defs}

%
\paragraph{Game Definition.}
\label{section:gameDef}
Game participants are referred as workers and master. In order to define the game played in each case, we follow the
customary notation used in Game Theory.
Given that this notation is repeatedly used throughout the paper, we summarize it in Table~\ref{table:gamenotation} for clarity.
We assume that the master always chooses an odd number of workers $n$, which avoids ties in voting settings where the answer space is binary as it is assumed in this paper as a worst case.
In order to model 
\cgaa{collusion} among workers, we view the set of workers as a set of \mma{non-empty} subsets $W=\{W_1,\dots,W_{\ell}\}$ such that $\sum_{i=1}^{\ell} |W_i|=n$ and $W_i\cap W_j=\emptyset$ for all $i\neq j$, $1\leq i,j\leq \ell$. 
We refer to each of these subsets as a \emph{group of workers} or a \emph{group} for short. 
We also refer to groups and the master as \emph{players}.
Workers in the same group act homogeneously, i.e., either all choose to cheat, or all choose to be honest, \mma{perhaps randomizing their decision by} tossing a unique coin. 
Workers acting individually are modeled as a group of size one. 
It is assumed that the size or composition of each group is known only to the members of the group,
\cgaa{but all cheating groups return the same incorrect answer.}

\begin{table*}[p]\centering
\begin{tabular}{|c|l|}
\hline
$m_s$& payoff of the master for the strategy profile $s$\\
\hline
$M$& master processor\\
\hline
$p_{s_i}^{(i)}$& probability that group $W_i$ uses strategy ${s_i}$\\
\hline
$p_{s_M}$& probability that the master uses strategy ${s_M}$\\
\hline
$s$& strategy profile (a mapping from players to pure strategies)\\
\hline
$s_{-i}$& strategy used by each \mma{player but $W_i$} in the strategy profile $s$\\
\hline
$s_i$& strategy used by \mma{group $W_i$} in the strategy profile $s$\\
\hline
${S}_i=\{{\CH},\overline{{\CH}}\}$& set of pure strategies (cheat/not-cheat) available to group $W_i$\\
\hline
$s_{-M}$& strategy used by each \mma{player but the master} in the strategy profile $s$\\
\hline
$s_M$& strategy used by \mma{the master} in the strategy profile $s$\\
\hline
${S}_M=\{{\VRF},\overline{{\VRF}}\}$& set of pure strategies (verify/not-verify) of the master\\
\hline
$supp(\sigma_i)$& set of strategies of \mma{group $W_i$} with probability $>0$ (called support) in $\sigma$\\
\hline
$supp(\sigma_M)$& set of strategies of \mma{the master} with probability $>0$  (called support) in $\sigma$\\
\hline
$\sigma$& mixed strategy profile (a mapping from players to prob. distrib. over pure strategies)\\
\hline
$\sigma_{-i}$& probability distribution over pure strategies used by each \mma{player but $W_i$} in $\sigma$\\
\hline
$\sigma_i$& probability distribution over pure strategies used by \mma{group $W_i$} in $\sigma$\\
\hline
$\sigma_{-M}$& probability distribution over pure strategies used by \mma{each player but the master} in $\sigma$\\
\hline
$\sigma_M$& probability distribution over pure strategies used by \mma{the master} in $\sigma$\\
\hline
$U_i(s_i,\sigma_{-i})$& expected utility of group $W_i$ with mixed strategy profile $\sigma$\\ 
\hline
$U_M(s_M,\sigma_{-M})$& expected utility of master with mixed strategy profile $\sigma$\\ 
\hline
$w_s^{(i)}$& payoff of group $W_i$ for the strategy profile $s$\\
\hline
$W=\{W_1,\dots,W_{\ell}\}$& set of \mma{worker groups}\\
\hline
\end{tabular}
\caption{Game notation}
\label{table:gamenotation}
\end{table*}


A strategy profile is defined as a mapping from \mma{players} to pure strategies, denoted as $s$. 
For succinctness, we express a strategy profile as a collection of individual strategy choices together with \mma{collective} strategy choices. For instance, $s_i={\CH},s_M={\VRF},R_{-iM},F_{-iM},T_{-iM}$ stands for a strategy profile $s$ where 
 group $W_i$ chooses strategy ${\CH}$ (to cheat), 
 the master chooses strategy ${\VRF}$ (to verify), 
 a set $R_{-iM}$ of groups (where group $W_i$ and the master are not included) randomize their strategy choice with probability $p_{\CH}\in (0,1)$, 
 a set $F_{-iM}$ of groups \mma{deterministically} choose  strategy ${\CH}$, 
 and a set $T_{-iM}$ of groups \mma{deterministically} choose strategy $\overline{{\CH}}$ (to be honest). 
For games with one worker and the master, the strategy profile is composed only by their choices. For example, $m_{\CH\VRF}$ stands for the master's payoff in the case that the worker cheated and the master verified.
We require that, for each group $W_i$, $p_{\CH}^{(i)}=1-p_{\overline{\CH}}^{(i)}$ and, for the master, $p_{\VRF}=1-p_{\overline{\VRF}}$.
For games where we only have one group or all groups
 use the same probability, we will express $p_{\CH}^{(i)}$ (resp. $p_{\overline{{\CH}}}^{(i)}$) simply by $p_{\CH}$ (resp. $p_{\overline{{\CH}}}$).
Whenever the strategy is clear from the context, we will refer to the expected utility of group $W_i$ as $U_i$, and for the master as $U_M$.
In the games studied the master and the workers have complete information on 
the algorithm and the parameters involved, except on the number and
the composition of the colluding groups.

\paragraph{Equilibrium Definition.}
We define now precisely the conditions for the equilibrium.
In this context, the probability distributions are not independent among members of a group. Furthermore, 
the formulation of equilibrium conditions among individual workers would violate the very definition of equilibrium since the choice of a worker does change the choices of other workers. Instead, equilibrium conditions are formulated among groups. 
Of course, the computation of an equilibrium might not be possible since the size of the groups is unknown. 
But, finding appropriate conditions so that the unique equilibrium is the same independently of that size, the problem may be solved.
As it will be seen in the general analysis, depending on the specific combination of payoffs, reward models, and games, knowing some bound (e.g. the trivial one) on the size of the smallest and/or largest group is enough, and sometimes not even necessary. Furthermore, as shown in Section~\ref{sec:MD}, there are cases where all groups are singleton because non-redundant allocation is the best strategy.
An important point to be made is that the majority is evaluated in terms of number of single answers. Nevertheless, this fact has an impact on the payoffs of each player, which in this case is a whole group, but not in the correctness of the equilibrium formulation.

\cgaa{
Recall from~\cite{osborne} that for any finite game, a mixed strategy profile $\sigma^*$
is a \emph{mixed-strategy Nash equilibrium} (MSNE) if, and only if, for each player $\pi$ (either
a worker group or the master),
\begin{align}
\label{MSNEeq}U_\pi(s_\pi,\sigma_{-\pi}^*) &= U_\pi(s'_\pi,\sigma_{-\pi}^*), \forall s_\pi,s'_\pi \in supp(\sigma_\pi^*),\\
\label{MSNEineq}U_\pi(s_\pi,\sigma_{-\pi}^*) &\geq U_\pi(s'_\pi,\sigma_{-\pi}^*),
\forall s_\pi,s'_\pi: s_\pi\in supp(\sigma_\pi^*),
s'_\pi\notin supp(\sigma_\pi^*).
\end{align}

In words, given a MSNE with mixed-strategy profile $\sigma^*$,  for each player $\pi$, the expected utility, assuming that all other players do not change their choice, is the same for each pure strategy that the player can choose with positive probability in $\sigma^*$, and it is not less than the expected utility of any pure strategy with probability zero of being chosen in $\sigma^*$.
A \emph{fully} MSNE is an equilibrium with mixed strategy profile $\sigma$ where, for each player $\pi$, $supp(\sigma_\pi)={S}_\pi$.}

\paragraph{Payoffs Definition.}

We detail in Table~\ref{table:payoffs} the payoff definitions that will be used throughout the paper.
All the parameters in this table are non-negative. 

\begin{table}[!h]\centering
\begin{small}
\begin{tabular}{|c|l|}
\hline
$\Cp$& worker's punishment for being caught cheating\\
\hline
$\Ct$& group's cost for computing the task\\
\hline
$\SW$& worker's benefit from master's acceptance\\
\hline
$\Cw$& master's punishment for accepting a wrong answer\\
\hline
$\SM$& master's cost for accepting the worker's answer\\
\hline
$\Cv$& master's cost for verifying worker's answers\\
\hline
$\Bc$& master's benefit from accepting the right answer\\
\hline
\end{tabular}
\end{small}
\caption{Payoffs}
\label{table:payoffs}
\end{table}

Notice that
we split the reward to a worker into $\SW$ and $\SM$, to model the fact
that the cost of the master might be different than the benefit of a
worker. In fact, in some models they may be completely unrelated.
Among the parameters involved, we assume that the master has the freedom of choosing the cheater penalty $\Cp$ and the worker reward for computing $\SM$. By tuning these parameters and choosing $n$, the master achieves the desired trade-off between correctness and cost.
Given that the master does not know the composition of groups (if there is any), benefits and punishments are applied individually to each worker, except for the cost for computing the task $\Ct$ which is shared among all workers belonging to the same group \cga{(as it was explained in the Introduction)}.
Sharing the task cost while being paid/punished individually may provide incentive to collude, but it models precisely the real world situation where the collusion is carried out in secret. Otherwise, a colluding group could be simply taken as a single player.



\section{Equilibria Analysis}
\label{sec:Gs}

In the following sections, different games are studied 
depending on the participants involved.
In order to identify the parameter conditions for which there is an NE, Equations~(\ref{MSNEeq}) and~(\ref{MSNEineq}) 
of the MSNE definition are instantiated in each particular game, without making any assumptions on the payoffs. We call this {\bf\em the general payoffs model}. From these instantiations, we obtain conditions on the parameters (payoffs and probabilities) that would make such equilibrium unique.
Finally, we introduce the reward models described before on those conditions, so that we can compare among all games and models in Section~\ref{sec:MD}.

\subsection{Game \gameone: One Master - One Worker}
\label{sec:G1}

We start the analysis by considering the game between the master and only one worker. Hence, collusions can not occur and we refer to the group just as ``the worker."

\paragraph{General Payoffs Model.}
\label{section:1M1Wgeneral}
In order to evaluate all possible equilibria, all the different mixes have to be considered. In other words, according with the range of values that $p_{\CH}$ and $p_{\VRF}$ can take, we can have fully MSNE, partially MSNE, or pure-strategies NE.
More specifically, both $p_{\CH}$ and $p_{\VRF}$ can take values either $0$, $1$, or in the open interval $(0,1)$. Depending on these values, the different conditions in Equations~(\ref{MSNEeq}) and~(\ref{MSNEineq}) 
have to be achieved in order to have an equilibrium.
Hence, conditions on $p_{\CH}$ and $p_{\VRF}$ for each equilibrium can be obtained from these equations.

For instance, for the case when $p_{\CH}\in(0,1),p_{\VRF}\in(0,1)$: From Equation~(\ref{MSNEeq}), 
 there is a fully MSNE if $U_M(\VRF ,p_{\CH})=U_M(\overline{\VRF },p_{\CH})$ and
 $U_W({\CH},p_{\VRF})=U_W(\overline{\CH},p_{\VRF})$ simultaneously. These equations
 determine the value of $p_{\CH}$ and $p_{\VRF}$ in the MSNE as follows.
\begin{gather*}
p_{\CH} m_{{\CH}\VRF } + (1-p_{\CH}) m_{\overline{\CH}\VRF }=
p_{\CH} m_{\CH\overline{\VRF }} + (1-p_{\CH}) m_{\overline{\CH}\overline{\VRF }}\\
p_{\CH} = \frac{m_{\overline{\CH}\overline{\VRF }} - m_{\overline{\CH}\VRF }}{m_{{\CH}\VRF } - m_{\overline{\CH}\VRF }- m_{\CH\overline{\VRF }} + m_{\overline{\CH}\overline{\VRF }}}.
\end{gather*}
\begin{gather*}
p_{\VRF} w_{{\CH}\VRF }+(1-p_{\VRF})w_{\CH \overline{\VRF }}=
p_{\VRF} w_{\overline{\CH }\VRF }+(1-p_{\VRF})w_{\overline{\CH }\overline{\VRF }}\\
p_{\VRF}=\frac{w_{\overline{\CH }\overline{\VRF }}-w_{\CH \overline{\VRF }}}{w_{\CH \VRF }-w_{\CH \overline{\VRF }}-w_{\overline{\CH }\VRF }+w_{\overline{\CH }\overline{\VRF }}}.
\end{gather*}
The computation of conditions for the other range-cases of $p_{\CH}$ and $p_{\VRF}$ is similar. 

On the other hand, the expected utility of the master and the worker in any equilibrium are
$U_M=p_{\CH}p_{\VRF} m_{{\CH}{\VRF}} + (1-p_{\CH})p_{\VRF} m_{\overline{{\CH}}{\VRF}}+
p_{\CH}(1-p_{\VRF}) m_{{\CH}\overline{{\VRF}}} + (1-p_{\CH})(1-p_{\VRF}) m_{\overline{{\CH}}\overline{{\VRF}}}$ and
$U_W=p_{\CH}p_{\VRF} w_{{\CH}{\VRF}} + p_{\CH}(1-p_{\VRF}) w_{{\CH}\overline{{\VRF}}}+
(1-p_{\CH})p_{\VRF} w_{\overline{{\CH}}{\VRF}} + (1-p_{\CH})(1-p_{\VRF}) w_{\overline{{\CH}}\overline{{\VRF}}}$ respectively, and the probability of accepting the wrong answer is
$\mathbf{P}_{wrong}=(1-p_{\VRF})p_{\CH}$.

\paragraph{Reward Model \modelone.}
\label{section:paytomajority11}
Recall that in this model we assume that when the master does not verify, it rewards only the majority. Given that there is only one worker, in this case the master rewards always. Under the payoff model detailed in Table~\ref{table:payoffs}, the payoffs are
\begin{center}
\begin{tabular}{ll}
$m_{{\CH}{\VRF}}=-\Cv$&
$w_{{\CH}{\VRF}}=-\Cp$\\
$m_{\overline{{\CH}}{\VRF}}=\Bc-\Cv-\SM$&
$w_{\overline{{\CH}}{\VRF}}=\SW-\Ct$\\
$m_{{\CH}\overline{{\VRF}}}=-\Cw-\SM$&
$w_{{\CH}\overline{{\VRF}}}=\SW$\\
$m_{\overline{{\CH}}\overline{{\VRF}}}=\Bc-\SM$&
$w_{\overline{{\CH}}\overline{{\VRF}}}=\SW-\Ct$
\end{tabular}
\end{center}
Replacing appropriately, we obtain the conditions for equilibrium,
probability of accepting the wrong answer, and utilities for each case.


\paragraph{Reward Model \modeltwo.}
In this model we assume that if the master does not
verify, it rewards all workers independently of the answer. Hence, the analysis is identical to the
previous case.


\paragraph{Reward Model \modelthree.}
\label{section:paytonone11}
Recall that in this model we assume that if the master does not
verify, it does not reward the worker. Hence, under the payoff
model detailed in Table~\ref{table:payoffs}, the payoffs are:
\begin{center}
\begin{tabular}{ll}
$m_{{\CH}{\VRF}}=-\Cv$&
$w_{{\CH}{\VRF}}=-\Cp$\\
$m_{\overline{{\CH}}{\VRF}}=\Bc-\Cv-\SM$&
$w_{\overline{{\CH}}{\VRF}}=\SW-\Ct$\\
$m_{{\CH}\overline{{\VRF}}}=-\Cw$&
$w_{{\CH}\overline{{\VRF}}}=0$\\
$m_{\overline{{\CH}}\overline{{\VRF}}}=\Bc$&
$w_{\overline{{\CH}}\overline{{\VRF}}}=-\Ct$
\end{tabular}
\end{center}


Replacing appropriately, we obtain the conditions for equilibrium,
probability of accepting the wrong answer, and utilities for each case, as we will see in the next section. The probability of
accepting the wrong result, the master utility for each case, the
conditions for equilibrium, and the workers utility for the reward models \modelone and \modelthree
can be obtained from
Tables~\ref{table:game2model1} and~\ref{table:game2model3}
by replacing $n=1$.

\subsection{Game \gametwo: $n$ Games One to One}
\label{sec:G2}

In this section it is considered the case where the master runs $n$ instances of the one to one game analyzed in the previous section. Workers are assumed to compute the equilibrium as if they were playing alone against the master. Hence, given the assumption that the players are rational and compute the equilibrium to decide what to do, the consideration of collusion is meaningless for this game. 
Hence, all groups are assumed to have exactly one member\cgaa{; we do assume however
that cheaters return the same incorrect value (to obtain worst case analysis).} Games where workers know about the existence of other workers and they can collude to fool the master are studied later.
Given the equilibria computed in Section~\ref{sec:G1}, the master runs
$n$ instances of that game, one with each of the $n$ workers, choosing
to verify or not with probability $p_{\VRF}$ only once.  Additionally, when
paying while not verifying, the master rewards all or none according
with the one-to-one game.

\noindent{\bf General Payoffs Model.}
Since this game is just a multiple-instance version of the previous game, under the payoff model detailed in Table~\ref{table:payoffs}, the conditions for equilibria and the utility of a worker are the same as in Section~\ref{sec:G1}.
However, the expected utility of the master and the probability of accepting
the wrong result change. In order to give those expressions, we
define the following notation. Let ${W}$ be the set of
partitions in two subsets $(F,T)$ of $W$, i.e.,
${W}=\{(F,T)|F\cap T=\emptyset,F\cup T=W\}$. $F$
is the set of workers that cheat and $T$ the set of honest workers. We also define master payoff functions
$m_s:\{0,1,\dots,n\}\rightarrow \mathbb{R}$, that still depend on the number
of workers that cheat or not, but are not necessarily just $n$ times the individual payoff of a \gameone game (reflecting the fact
that the cost may include some fixed amount for unique verification or
unique cost of being wrong). For the sake of clarity, we
will denote the probability that the majority cheats as
$\mathbf{P}_{\CH}$.


\begin{table}[p]\centering
\begin{scriptsize}
\begin{tabular}{|c|c|c|c|c|}
\hline
\begin{tabular}{c}
Equilibrium\\$p_{\CH}, p_{\VRF}$
\end{tabular}&
Conditions&$\mathbf{P}_{wrong}$&$U_M$&$U_{W_i}$\\
\hline
$\frac{\Cv}{\SM+\Cw}$,
$\frac{\Ct}{\SW+\Cp}$&
&
$(1-p_{\VRF})\mathbf{P}_{\CH}$&
\begin{minipage}{1cm}
\begin{tabbing}
$p_{\VRF}($\=$(1-p_{\CH}^n)\Bc-$\\
\>$\Cv-(1-p_{\CH})n\SM)+$\\
$(1-p_{\VRF})($\=$\Bc(1-\mathbf{P}_{\CH})-$\\
\>$\Cw\mathbf{P}_{\CH}-n\SM)$
\end{tabbing}
\end{minipage}
&
$\SW-\Ct$\\
\hline
$0$,
\begin{tabular}{c}
$\frac{\Ct}{\SW+\Cp}\leq p_{\VRF}<1$\\
$0<p_{\VRF}$
\end{tabular}&
$\Cv=0$&
$0$&
$\Bc-n\SM$&
$\SW-\Ct$\\
\hline
$1$,
\begin{tabular}{c}
$0<p_{\VRF}\leq\frac{\Ct}{\SW+\Cp}$\\
$p_{\VRF}<1$
\end{tabular}&
$\Cv=\Cw+\SM$&
$1-p_{\VRF}$&
\begin{tabular}{c}
$-p_{\VRF}\Cv-(1-p_{\VRF})(\Cw+n\SM)$
\end{tabular}&
\begin{tabular}{c}
$(1\!-\!p_{\VRF})\SW\!-\!\!\!$\\$p_{\VRF}\Cp$
\end{tabular}\\
\hline
\begin{tabular}{c}
$0\leq p_{\CH}\leq\frac{\Cv}{\SM+\Cw}$\\
$p_{\CH}<1$
\end{tabular},
$0$&
$\Ct=0$&
$\mathbf{P}_{\CH}$&
$\Bc(1-\mathbf{P}_{\CH})-\Cw\mathbf{P}_{\CH}-n\SM$&
$\SW$\\
\hline
\begin{tabular}{c}
$\frac{\Cv}{\SM+\Cw}\leq p_{\CH}<1$\\
$0<p_{\CH}$
\end{tabular},
$1$&
$\Ct=\SW+\Cp$&
$0$&
\begin{tabular}{l}
$(1-\prod_{j\in W} p_{\CH}^{(j)})\Bc-\Cv-$\\
$\sum_{(W_F,W_T)\in{W}}$
$\prod_{j\in W_F} p_{\CH}^{(j)}\cdot$\\
$\prod_{k\in W_T} (1-p_{\CH}^{(k)}) |W_T|\SM$
\end{tabular}&
$-\Cp$\\
\hline
1,
1&
\begin{tabular}{c}
$\Cv\leq \Cw+\SM$\\
$\Ct\geq \SW+\Cp$
\end{tabular}&
$0$&
$-\Cv$&
$-\Cp$\\
\hline
0,
1&
\begin{tabular}{c}
$\Cv=0$\\
$\Ct\leq \SW+\Cp$
\end{tabular}&
$0$&
$\Bc-n\SM$&
$\SW-\Ct$\\
\hline
1,
0&
$\Cv \geq \Cw+\SM$&
$1$&
$-\Cw-n\SM$&
$\SW$\\
\hline
\end{tabular}
\end{scriptsize}
\caption{Game \gametwo, Models \modelone and \modeltwo (and Game \gameone for $n=1$)\vspace{-5em}}
\label{table:game2model1}
\end{table}


\begin{table}[p]\centering
\begin{scriptsize}
\begin{tabular}{|c|c|c|c|c|}
\hline
\begin{tabular}{c}
Equilibrium\\
$p_{\CH}, p_{\VRF}$
\end{tabular}&Conditions&$\mathbf{P}_{wrong}$&$U_M$&$U_{W_i}$\\
\hline
$\frac{\Cv+\SM}{\SM+\Cw}$,
$\frac{\Ct}{\SW+\Cp}$&
&
$(1-p_{\VRF})\mathbf{P}_{\CH}$&
\begin{minipage}{1cm}
\begin{tabbing}
$p_{\VRF}($\=$(1-p_{\CH}^n)\Bc-$\\
\>$\Cv-(1-p_{\CH})n\SM)+$\\
$(1-p_{\VRF})($\=$\Bc(1-\mathbf{P}_{\CH})-$\\
\>$\Cw\mathbf{P}_{\CH})$
\end{tabbing}
\end{minipage}
&
$-p_{\VRF}\Cp$\\
\hline
$0$,
\begin{tabular}{c}
$\frac{\Ct}{\SW+\Cp}\leq p_{\VRF}<1$\\
$0<p_{\VRF}$
\end{tabular}&
$\SM=\Cv=0$&
$0$&
$\Bc$&
$p_{\VRF}\SW-\Ct$\\
\hline
$1$,
\begin{tabular}{c}
$0<p_{\VRF}\leq\frac{\Ct}{\SW+\Cp}$\\
$p_{\VRF}<1$
\end{tabular}&
$\Cv=\Cw$&
$1-p_{\VRF}$&
$-\Cv$&
$-p_{\VRF}\Cp$\\
\hline
\begin{tabular}{c}
$0\leq p_{\CH}\leq\frac{\Cv+\SM}{\SM+\Cw}$\\
$p_{\CH}<1$
\end{tabular},
$0$&
$\Ct=0$&
$\mathbf{P}_{\CH}$&
$\Bc(1-\mathbf{P}_{\CH})-\Cw\mathbf{P}_{\CH}$&
$0$\\
\hline
\begin{tabular}{c}
$\frac{\Cv+\SM}{\SM+\Cw}\leq p_{\CH}<1$\\
$0<p_{\CH}$
\end{tabular},
$1$&
$\Ct=\SW+\Cp$&
$0$&
\begin{tabular}{l}
$(1-\prod_{j\in W} p_{\CH}^{(j)})\Bc-\Cv-$\\
$\sum_{(W_F,W_T)\in{W}}$
$\prod_{j\in W_F} p_{\CH}^{(j)}\cdot$\\
$\prod_{k\in W_T} (1-p_{\CH}^{(k)}) |W_T|\SM$
\end{tabular}&
$-\Cp$\\
\hline
$1$,
$1$&
\begin{tabular}{c}
$\Cv\leq \Cw$\\
$\Ct\geq \SW+\Cp$
\end{tabular}&
$0$&
$-\Cv$&
$-\Cp$\\
\hline
$0$,
$1$&
\begin{tabular}{c}
$\Cv=\SM=0$\\
$\Ct\leq \SW+\Cp$
\end{tabular}&
$0$&
$\Bc$&
$\SW-\Ct$\\
\hline
$1$,
$0$&
$\Cv\geq \Cw$&
$1$&
$-\Cw$&
$0$\\
\hline
\end{tabular}
\end{scriptsize}
\caption{Game \gametwo, Model \modelthree (and Game \gameone for $n=1$)}
\label{table:game2model3}
\end{table}

Then, the probability that the majority cheats, the probability of being wrong, and the master's utility are
\begin{align*}
\mathbf{P}_{\CH}=&\sum_{\substack{(F,T)\in{W}\\|F|>|T|}} \prod_{f\in F} p_{\CH}^{(f)} \prod_{t\in T} (1-p_{\CH}^{(t)}),\\
\mathbf{P}_{wrong}=&(1-p_{\VRF})\mathbf{P}_{\CH},\\
U_M=&p_{\VRF}\sum_{(F,T)\in{W}} \prod_{f\in F} p_{\CH}^{(f)} \prod_{t\in T} (1-p_{\CH}^{(t)}) m_{\VRF}+
(1-p_{\VRF})\sum_{(F,T)\in{W}} \prod_{f\in F} p_{\CH}^{(f)} \prod_{t\in T} (1-p_{\CH}^{(t)}) m_{\overline{{\VRF}}}.
\end{align*}
Respectively,
where $m_{\VRF}=m_{{\CH}{\VRF}}(|F|)+m_{\overline{{\CH}}{\VRF}}(|T|)$ and $m_{\overline{{\VRF}}}=m_{{\CH}\overline{{\VRF}}}(|F|)+m_{\overline{{\CH}}\overline{{\VRF}}}(|T|)$.


\paragraph{Reward Models.}
In this game, we assume that the cost of verification $\Cv$ is independent of the number of workers (since all cheating workers return the same value) and that, as long as some
worker is honest, upon verification the master obtains the
correct result. It is important to note that, under this assumption,
the probability of obtaining the correct result is not
$1-\mathbf{P}_{wrong}$, given that if the master verifies but all
workers cheat, the master does not obtain the correct result. Recall that
the master plays $n$ instances of a one-to-one game, thus, depending
on the model, it must reward every worker if not verifying
independently of majorities.  We summarize the probability of
accepting the wrong result, the master utility for each case, the
conditions for equilibrium, and the workers utility for the reward models \modelone and \modelthree
in Tables~\ref{table:game2model1} and~\ref{table:game2model3}
respectively (Tables~\ref{table:game2model1}
and~\ref{table:game2model3} give also these values for Game \gameone
replacing appropriately $n=1$).

\subsection{Game \gamethree: No Master in the Game}
\label{sec:G4}
Another natural generalization of the game of Section~\ref{sec:G1} is
to consider a game in which the master assigns the task to $n$ workers that play the game among them. Intuitively, it can be seen that, in case of not verifying, workers will compete to be in the majority (to persuade the master). 
Given that workers know the existence of the other workers, including collusions in the analysis is in order.
The question of how the participation of the master in the game would affect the results obtained in this section is addressed in Section~\ref{sec:G3}.

\paragraph{General Payoffs Model.}
In order to analyze this game, it is convenient to partition the set of groups. More precisely, consider disjoint sets $F$, $T$ and $R$, such that $F\cup T\cup R=W$, as follows. 
$F$ is the set of groups that choose to cheat as a pure strategy, i.e., $F=\{W_i|W_i\in W\land p_{\CH}^{(i)}=1\}$. 
$T$ is the set of groups that choose not to cheat as a pure strategy, i.e., $T=\{W_i|W_i\in W\land p_{\CH}^{(i)}=0\}$. 
$R$ is the set of groups that randomize their choice, i.e., $R=\{W_i|W_i\in W\land p_{\CH}^{(i)}\in(0,1)\}$. 
Let $F_{-i}=F\setminus \{W_i\}$, $T_{-i}=T\setminus \{W_i\}$, and $R_{-i}=R\setminus \{W_i\}$.
Let $\Gamma_{-i}$ be the set of partitions in two subsets $(R_F,R_T)$ of $R_{-i}$, i.e., $\Gamma_{-i}=\{(R_F,R_T)|R_F\cap R_T=\emptyset\land R_F\cup R_T=R_{-i}\}$. 
Let $\mathbf{E}[w_{s}^{(i)}]$ be the expected payoff of group $W_i$ for the strategy profile $s$, taking the expectation over the choice of the master of verifying or not. 
Then, for each group $W_i\in W$ and for each strategy profile $s_{-i}=R_{-i},F_{-i},T_{-i}$, we have
\begin{align*}
U_i&(s_{-i},s_i={\CH})=
\sum_{(R_F,R_T)\in\Gamma_{-i}} \prod_{W_f\in R_F} p_{\CH}^{(f)}
 \prod_{W_t\in R_T} (1-p_{\CH}^{(t)}) \mathbf{E}[w_{\substack{F_{-i}\cup R_F,\\T_{-i}\cup
     R_T,\\s_i={\CH}}}^{(i)}],\\
U_i&(s_{-i},s_i=\overline{{\CH}})=
\sum_{(R_F,R_T)\in\Gamma_{-i}}
 \prod_{W_f\in R_F} p_{\CH}^{(f)} \prod_{W_t\in R_T} (1-p_{\CH}^{(t)})
 \mathbf{E}[w_{\substack{F_{-i}\cup R_F,\\T_{-i}\cup R_T,\\s_i=\overline{{\CH}}}}^{(i)}].
 \end{align*}

In words, the expected utility of a worker in a group that chooses to cheat (resp. to be honest) is, by linearity of expectation, the expected payoff of the worker, the expectation taken over the choice of the master, averaged over all combinations of outcomes cheat/honest of other groups that choose to randomize their strategy choice, this average weighted by the probability of such outcomes.

In order to find conditions for a desired equilibrium, we study 
what we call the \emph{utility differential} of a worker,  i.e. the difference on the expected utility of a worker if its group chooses to cheat with respect to the case when the group chooses to be honest. More precisely,
\begin{align} 
\Delta U_i(s)=U_i(s_{-i},s_i=\CH)-U_i(s_{-i},s_i=\overline{\CH}).\label{diffutil}
\end{align}

\mmaa{


For clarity, define $N_{F-i}=\sum_{S\in F_{-i}\cup R_F} |S|$ and $N_{T-i}=\sum_{S\in T_{-i}\cup R_T} |S|$,
i.e. the number of cheaters and honest workers respectively except for those in group $W_i$.
We also define what we call the \emph{payoff differential} as the difference on the expected payoff of a worker, the expectation taken over the choice of the master, if its group chooses to cheat with respect to the case when the group chooses to be honest. Furthermore, we denote the payoff differential depending on whether the size of the group has an impact on what is the majority outcome. 
More precisely, for each partition $(R_F,R_T)\in\Gamma_i$, let

$\Delta w_{\CH}^{(i)}=\mathbf{E}[w_{s_i={\CH}}^{(i)}]-\mathbf{E}[w_{s_i=\overline{{\CH}}}^{(i)}]$, 
when $N_{F-i}-N_{T-i}>|W_i|$,

$\Delta w_{\overline{\CH}}^{(i)}=\mathbf{E}[w_{s_i={\CH}}^{(i)}]-\mathbf{E}[w_{s_i=\overline{{\CH}}}^{(i)}]$,
when $N_{T-i}-N_{F-i}>|W_i|$, and

$\Delta w_X^{(i)}=\mathbf{E}[w_{s_i={\CH}}^{(i)}]-\mathbf{E}[w_{s_i=\overline{{\CH}}}^{(i)}]$,
when $|N_{F-i}-N_{T-i}|<|W_i|$.

In words, the payoff differential of each worker in a group when the majority cheats or is honest independently of the group's choice, and the payoff differential when the decision of the group may change the majority.
Given that the payoff depends only on the outcome majority, replacing this notation in Equation~\ref{diffutil}, we have

\begin{align}
\Delta U_i(s)=\nonumber
&\Delta w_{\CH}^{(i)}
\sum\!\!\!_{\substack{(R_F,R_T)\in\Gamma_{-i}\\N_{F-i}-N_{T-i}>|W_i|}}
\prod_{W_f\in R_F} p_{\CH}^{(f)}
\prod_{W_t\in R_T} (1-p_{\CH}^{(t)})+\nonumber\\
&\Delta w_X^{(i)}
\sum\!\!\!_{\substack{(R_F,R_T)\in\Gamma_{-i}\\|N_{F-i}-N_{T-i}|<|W_i|}}
\prod_{W_f\in R_F} p_{\CH}^{(f)}
\prod_{W_t\in R_T} (1-p_{\CH}^{(t)})+\nonumber\\
&\Delta w_{\overline{{\CH}}}^{(i)}
\sum\!\!\!_{\substack{(R_F,R_T)\in\Gamma_{-i}\\N_{T-i}-N_{F-i}>|W_i|}}
\prod_{W_f\in R_F} p_{\CH}^{(f)}
\prod_{W_t\in R_T} (1-p_{\CH}^{(t)}).
\label{eq:G3grp}
\end{align}
}

In words, the utility differential of a worker is the average of its payoff differential over the three cases defined by the impact of its group over the majority, this average weighted by the probability of such cases.

Restating Equations~(\ref{MSNEeq}) or~(\ref{MSNEineq}) in terms of Equation~(\ref{eq:G3grp}), the equilibrium conditions are, 
for each group that does not choose a pure strategy, the differential utility must be zero ($\forall i\in R, \Delta U_i(s)=0$);
for each group that chooses to cheat as a pure strategy, the differential utility must not be negative ($\forall i\in F,\Delta U_i(s)\geq 0$); 
and for each group that chooses to be honest as a pure strategy, the differential utility must not be positive ($\forall i\in T,\Delta U_i(s)\leq 0$).

The following lemma, which is crucially used in the rest of our analysis, shows that, if there is a given total order among the payoff differentials defined, in order to attain a unique equilibrium all groups must decide deterministically. The proof is based on an algebraic argument. 

\begin{lemma}
\label{lemma:nouniquemixedgroups}
Given a game as defined, if $\Delta w_{\CH}^{(i)}\geq\Delta w_X^{(i)}\geq\Delta w_{\overline{\CH}}^{(i)}$ for every group $W_i\in W$, then there is no unique equilibrium where $R\neq\emptyset$ (i.e, all groups decide deterministically).
\end{lemma}
\begin{proof}
 For the sake of contradiction, assume there is a unique equilibrium
 $\sigma$ for which $R\neq\emptyset$ and $\Delta w_{\CH}^{(i)}\geq\Delta
 w_X^{(i)}\geq\Delta w_{\overline{\CH}}^{(i)}$ for every group $W_i\in W$.
 Then, for every group $W_i\in R$, $\Delta U_i(s)=0$ must be solvable.
 If $\Delta w_{\CH}^{(i)}\geq 0$, for all $W_i\in R$,
 there would be also an equilibrium where all groups in $R$ choose to cheat
 and $\sigma$ would not be unique, which is a contradiction. 
 Consider now the case where there exists some $W_i\in R$ such that $\Delta
 w_{\CH}^{(i)}<0$. Then, it must hold that $|R|>1$, otherwise $\Delta U_i=0$ is false for $W_i$.
 Given that $|R|>1$, the probabilities given by the summations in Equation~(\ref{eq:G3grp}) for $W_i$ are all
 strictly bigger than zero. Therefore, given that $\Delta U_i=0$ must be solvable, at least one of $\Delta
 w_X^{(i)}> 0$ and $\Delta w_{\overline{\CH}}^{(i)}> 0$ must hold, which is also a
 contradiction with the assumption that $\Delta w_{\CH}^{(i)}\geq\Delta
 w_X^{(i)}\geq\Delta w_{\overline{\CH}}^{(i)}$. 
\end{proof}

In the following sections, conditions to obtain unique
equilibria under different payoff models are studied. In all these models it holds that
$\Delta w_{\CH}^{(i)}\geq\Delta w_X^{(i)}\geq\Delta w_{\overline{\CH}}^{(i)}$ for all $W_i\in W$. Then, by
Lemma~\ref{lemma:nouniquemixedgroups}, there is no unique equilibrium where
$R\neq\emptyset$. Regarding equilibria where $R=\emptyset$, unless the task
assigned has a binary output \mma{(the answer can be negated)}, a unique equilibrium where all groups
choose to cheat is not useful. Then, we set up $p_{\VRF}$ so that $\Delta w_{\CH}^{(i)}<0$, $\Delta
w_X^{(i)}<0$ and $\Delta w_{\overline{\CH}}^{(i)}<0$ for all $W_i\in W$ so that $\Delta U_i\geq 0$ has
no solution and no group can choose to cheat as a pure strategy. Thus,
the only equilibrium is for all the groups to choose to be honest, which solves $\Delta U_i\leq 0$. Therefore, $p_{\CH}^{(i)}=0$, $\forall W_i\in W$, and hence $\mathbf{P}_{wrong}=0$.


\paragraph{Reward Model \modelone.}
\label{section:wWgrppaymaj}
Replacing appropriately the payoffs detailed in Table~\ref{table:payoffs}, we obtain for any group $W_i\in W$
\begin{align*}
\Delta w_{\CH}^{(i)} 
&= -p_{\VRF}|W_i|(\Cp+2\SW)+|W_i|\SW+\Ct,\\
\Delta w_X^{(i)} 
&= -p_{\VRF}|W_i|(\Cp+\SW)+\Ct,\\
\Delta w_{\overline{\CH}}^{(i)} 
&= -p_{\VRF}|W_i|\Cp-|W_i|\SW+\Ct.
\end{align*}

To make $\Delta w_{\CH}^{(i)}<0$ we want 
\begin{align*}
p_{\VRF}>\frac{|W_i|\SW+\Ct}{|W_i|(\Cp+2\SW)}, \forall W_i \in W.
\end{align*}

And the expected utilities are then
\begin{align*}
U_M &=\Bc-p_{\VRF}\Cv-n\SM\\
U_{W_i} &=|W_i|\SW-\Ct, \textrm{\ for each\ } W_i\in W.
\end{align*}

\paragraph{Reward Model \modeltwo.}
\label{section:wWgrppayall}
Similarly, for any group $W_i\in W$,
\begin{align*}
\Delta w_{\CH}^{(i)} 
&= -p_{\VRF}|W_i|(\Cp+\SW)+\Ct,\\
\Delta w_X^{(i)} 
&= -p_{\VRF}|W_i|(\Cp+\SW)+\Ct,\\
\Delta w_{\overline{{\CH}}}^{(i)} 
&= -p_{\VRF}|W_i|(\Cp+\SW)+\Ct.
\end{align*}

Then, the condition to obtain the desired unique equilibrium and the expected utilities are
\begin{align*}
p_{\VRF}&>\frac{\Ct}{|W_i|(\Cp+\SW)},\forall W_i\in W,\\
U_M&=\Bc-p_{\VRF}\Cv-n\SM,\\
U_{W_i}&=|W_i|\SW-\Ct, \textrm{\ for each\ } W_i\in W.
\end{align*}


\paragraph{Reward Model \modelthree.}
\label{section:wWgrppaynone}
Again, for any group $W_i\in W$,
\begin{align*}
\Delta w_{\CH}^{(i)} 
&= -p_{\VRF}|W_i|(\Cp+\SW)+\Ct,\\
\Delta w_X^{(i)} 
&= -p_{\VRF}|W_i|(\Cp+\SW)+\Ct,\\
\Delta w_{\overline{{\CH}}}^{(i)} 
&= -p_{\VRF}|W_i|(\Cp+\SW)+\Ct.
\end{align*}

And the condition to obtain the unique equilibrium and the expected utilities are
\begin{align*}
p_{\VRF}&>\frac{\Ct}{|W_i|(\Cp+\SW)},\forall W_i\in W,\\
U_M&=\Bc-p_{\VRF}(\Cv+n\SM),\\
U_{W_i}&=p_{\VRF}|W_i|\SW-\Ct, \textrm{\ for each\ } W_i\in W.
\end{align*}


In order to maximize the master utility we would like to design games where $p_{\VRF}$ is small. Therefore, we look for a lower bound on $p_{\VRF}$. It is easy to see that, in all of the three payoff models, the worst case lower bound is given by the group of minimum size. Although at a first glance this fact seems counterintuitive, it is not surprising due to the following two reasons. On one hand, colluders are likely to be in the majority, but the unique equilibrium occurs when all workers are honest. On the other hand, the extra benefit that workers obtain by colluding is not against the master interest since it is just a saving in computation costs.


\subsection{Game \gamefour: One Master - $n$ Workers}
\label{sec:G3}
We now observe how the conditions obtained in the previous game are modified if the master also participates as a player. The equilibria analysis regarding groups follows the same lines as in Section~\ref{sec:G4}. However, now Equations~(\ref{MSNEeq}) and~(\ref{MSNEineq}) have to be applied to the master, as follows.

\paragraph{General Payoffs Model.}
Recall that $R$ is the set of groups that randomize their
choice. Let $\Gamma$ be the set of partitions in two subsets $(R_F,R_T)$ of
$R$, i.e., $\Gamma=\{(R_F,R_T)|R_F\cap R_T=\emptyset\land R_F\cup
R_T=R\}$.  Then, for the master,
\begin{align*}
U_M&(R,F,T,s_M={\VRF})=
\sum_{(R_F,R_T)\in\Gamma} \prod_{f\in R_F} p_{\CH}^{(f)}
 \prod_{t\in R_T} (1-p_{\CH}^{(t)}) m_{\substack{F\cup R_F,\\T\cup
     R_T,\\s_M={\VRF}}}
\end{align*}
\begin{align*}
U_M&(R,F,T,s_M=\overline{{\VRF}})=
\sum_{(R_F,R_T)\in\Gamma} \prod_{f\in R_F}
 p_{\CH}^{(f)} \prod_{t\in R_T} (1-p_{\CH}^{(t)}) m_{\substack{F\cup
     R_F,\\T\cup R_T,\\s_M=\overline{{\VRF}}}}.
\end{align*}
From Equation~(\ref{MSNEeq}), if $p_{\VRF}\in(0,1)$, the MSNE condition is
$$U_M(R,F,T,s_M={\VRF})=U_M(R,F,T,s_M=\overline{{\VRF}}).$$  
From Equation~(\ref{MSNEineq}), if $p_{\VRF}=0$ the condition is
$$U_M(R,F,T,s_M={\VRF})\leq U_M(R,F,T,s_M=\overline{{\VRF}}),$$ 
and if $p_{\VRF}=1$ the condition is 
$$U_M(R,F,T,s_M={\VRF})\geq U_M(R,F,T,s_M=\overline{{\VRF}}).$$

The MSNE conditions for groups are the same as in Section~\ref{sec:G4}. Hence, the conditions obtained for each of the reward models are the same. However, additional conditions are obtained from the master-utility conditions as follows. As in Section~\ref{sec:G4}, the desired unique MSNE occurs when $p_{\CH}=0$. Using that, in the master-utility conditions we get for the reward model \modelone that if $p_{\VRF}<1$, 
$\Bc-\Cv-n\SM=\Bc-n\SM$,
and if $p_{\VRF}=1$, 
$\Bc-\Cv-n\SM\geq \Bc-n\SM$. Therefore, in any case it must hold $\Cv=0$.
For the reward model \modeltwo,
the master-utility conditions give, if $p_{\VRF}<1$,  
$\Bc-\Cv-n\SM=\Bc-n\SM$ and if $p_{\VRF}=1$,
$\Bc-\Cv-n\SM\geq \Bc-n\SM$. Therefore, again, $\Cv=0$.
Finally, for the reward model \modelthree, the master-utility conditions give if $p_{\VRF}<1$,
$\Bc-\Cv-n\SM=\Bc$ and if $p_{\VRF}=1$, $\Bc-\Cv-n\SM\geq \Bc$. Therefore,
$\Cv=\SM=0$.
Hence, to achieve the goal of forcing the groups to be honest, \emph{in this game, verifying must be free for the master.}



\section{Algorithmic Mechanisms}
\label{sec:MD}
In this section two realistic scenarios in which the
master-worker model considered could be naturally applicable are proposed. For
these scenarios, we determine appropriate games and parameters to be
used by the master to maximize its benefit. 

The basic protocol (mechanism) used by the master to accept the
correct task result while maximizing its benefit is as follows:
Given the payoff parameters (these can either be fixed by the system or be chosen
by the master), the master sends the task (to be computed), the game to be played,
the probability of verification $p_\VRF$, and the payoff model to be used.
For computational reasons, the master also sends a certificate to the workers. The certificate includes the strategy that the workers must play to achieve the unique NE, together with the appropriate data to demonstrate this fact%
\footnote{The certificate is included only for cases where resource limitations preclude the worker from computing the unique equilibrium, but it is not related to distributions over public signals (as in a correlated equilibrium) since workers do not randomize their choice according to this certificate.}. More details for the use of the certificate are given in Section~\ref{section:complex}.

After receiving the replies from all workers, and independently of the distribution of the answers, the master processor chooses to verify the answers with the probability $p_\VRF$. If the answers were not verified it accepts the result of the majority. Then, it applies the corresponding reward model. The protocol is detailed in Algorithm~\ref{alg1}.

\begin{algorithm}[h]
\label{alg1}
\SetKwFor{Upon}{upon}{do}{endupon}
send (task, game, $p_{\VRF}$, payoff model ${R}$, certificate) to all workers\;
\Upon{receiving all answers}{
verify the answers with probability $p_{\VRF}$\;\label{step:ch}
\If{the answers were not verified}{accept the majority\;}
apply the reward model\;
}
\caption{Master algorithm}
\end{algorithm}

Hence, the master, given the payoff parameters, can determine the 
game and parameters (including the value of $p_\VRF$) to force
the workers into a unique NE, that would yield the correct task result (with high probability) while maximizing the master's benefit. Examples of specific parameters (including the value of $p_{\VRF}$) and games such that the master can achieve this are analyzed in the following subsections. 

\subsection{SETI-like Scenario}

The first scenario considered is a volunteer computing system such
as SETI@home, where users accept to donate part of their processors
idle time to collaborate in the computation of large tasks. In this
case, we assume that workers incur in no cost to perform the task, but
they obtain a benefit by being recognized as having performed it
(possibly in the form of prestige\cga{, e.g, by being included on SETI's top contributors
list}). Hence, we assume that $\SW>\Ct=0$. The
master incurs in a (possibly small) cost $\SM$ when rewarding a worker
(e.g., by advertising its participation in the project).  As assumed
in the general model, in this model the master may verify the values
returned by the workers, at a cost $\Cv>0$. We also assume that the
master obtains a benefit $\Bc>\SM$ if it accepts the correct result of
the task, and suffers a cost $\Cw>\Cv$ if it accepts an incorrect
value. 

Under these constraints, the equilibria for games \gameone and \gametwo collapse to
one single equilibrium point. Also, since game~\gamefour requires free verification ($\Cv=0$) for
the equilibrium to be unique, it cannot be used in this
scenario. The different applicable cases are summarized in
Table~\ref{table:seti}. In this table it can be observed that in games \gameone
and \gametwo the equilibrium is achieved with any value of $p_{\CH}$ in an
interval. The master has no way to force the specific value of $p_{\CH}$
that a worker uses within the interval. And, in particular, it cannot force
$p_{\CH}=0$ (i.e., $\mathbf{P}_{wrong}=0$). Additionally, looking at the
master utility, all games have $U_M< \Bc$. However, in game (\gamethree,\modelthree) the
master can make $U_M$ arbitrarily close to $\Bc$ by setting $p_{\VRF}$
arbitrarily small. (Notice that the utility of a worker will be
arbitrarily small likewise, but given that workers are volunteering this is not a problem.) 
{\em In conclusion, the game (\gamethree,\modelthree)
with $n=1$ ($|W|=|W_i|=1$) and very small $p_{\VRF}$ is the best choice in this scenario, since it
satisfies $\mathbf{P}_{wrong}=0$ and $U_M \approx \Bc$. We highlight this observation in the following theorem.}

\begin{theorem}
For any set of payoff parameters that can be characterized as the SETI scenario, 
in order to obtain the correct answer (with probability $1$), 
it is enough to assign the task to only one worker excluding the master from the game,
and verify with arbitrarily small probability. 
Additionally, if the worker is rewarded only when the result is verified to be correct, 
the mechanism yields a utility for the master that is almost optimal.
\end{theorem}

\begin{table*}[tb]\centering
\begin{scriptsize}
\begin{tabular}{|c|c|c|c|c|}
\hline
(Game,Model) & Equilibrium&$\mathbf{P}_{wrong}$&$U_M$&$U_{W_i}$\\
&$p_{\CH},p_{\VRF}$&&&\\
\hline
(\gameone,\modelone), (\gameone,\modeltwo)&
\begin{tabular}{c}
$0\leq p_{\CH}\leq\frac{\Cv}{\SM+\Cw}$,
$p_{\CH}<1$
\end{tabular},
$p_{\VRF}=0$&
$p_{\CH}$&
$\Bc-p_{\CH}(\Bc+\Cw)-\SM$&
$\SW$\\
\hline
\begin{tabular}{c}
(\gameone,\modelthree)
\end{tabular}&
\begin{tabular}{c}
$0\leq p_{\CH}\leq\frac{\Cv+\SM}{\SM+\Cw}$,
$p_{\CH}<1$
\end{tabular},
$p_{\VRF}=0$&
$p_{\CH}$&
$\Bc-p_{\CH}(\Bc+\Cw)$&
$0$\\
\hline
\begin{tabular}{c}
(\gametwo,\modelone), (\gametwo,\modeltwo)
\end{tabular}&
\begin{tabular}{c}
$0\leq p_{\CH}\leq\frac{\Cv}{\SM+\Cw}$,
$p_{\CH}<1$
\end{tabular},
$p_{\VRF}=0$&
$\mathbf{P}_{\CH}$&
$\Bc-\mathbf{P}_{\CH}(\Bc+\Cw)-n\SM$&
$\SW$\\
\hline
\begin{tabular}{c}
(\gametwo,\modelthree)
\end{tabular}&
\begin{tabular}{c}
$0\leq p_{\CH}\leq\frac{\Cv+\SM}{\SM+\Cw}$,
$p_{\CH}<1$
\end{tabular},
$p_{\VRF}=0$&
$\mathbf{P}_{\CH}$&
$\Bc-\mathbf{P}_{\CH}(\Bc+\Cw)$&
$0$\\
\hline
\begin{tabular}{c}
(\gamethree,\modelone)
\end{tabular}&
$p_{\CH}=0$,
$\frac{\SW}{\Cp+2\SW}<p_{\VRF}\leq 1$&
$0$&
$\Bc-p_{\VRF}\Cv-n\SM$&
$|W_i|\SW$\\
\hline
\begin{tabular}{c}
(\gamethree,\modeltwo)
\end{tabular}&
$p_{\CH}=0$,
$0<p_{\VRF}\leq 1$&
$0$&
$\Bc-p_{\VRF}\Cv-n\SM$&
$|W_i|\SW$\\
\hline
\begin{tabular}{c}
(\gamethree,\modelthree)
\end{tabular}&
$p_{\CH}=0$,
$0<p_{\VRF}\leq 1$&
$0$&
$\Bc-p_{\VRF}(\Cv+n\SM)$&
$p_{\VRF}|W_i|\SW$\\
\hline
\end{tabular}
\end{scriptsize}
\caption{SETI-like Scenario}
\label{table:seti}
\end{table*}

\begin{table*}[tb]\centering
\begin{scriptsize}
\begin{tabular}{|c|c|c|c|c|}
\hline
\begin{tabular}{c}
(Game,Model)
\end{tabular}&
Equilibrium&$\mathbf{P}_{wrong}$&$U_M$&$U_{W_i}$\\
&$p_{\CH},p_{\VRF}$&&&\\
\hline
\begin{tabular}{c}
(\gameone,\modelone), (\gameone,\modeltwo)
\end{tabular}&
$\frac{\Cv}{\SM+\Cw}$,
$\frac{\Ct}{\SW+\Cp}$&
$(1-p_{\VRF})p_{\CH}$&
$\Bc-p_{\CH}(\Bc+\Cw)-\SM$&
$\SW-\Ct$\\
\hline
\begin{tabular}{c}
(\gameone,\modelthree)
\end{tabular}&
$\frac{\Cv+\SM}{\SM+\Cw}$,
$\frac{\Ct}{\SW+\Cp}$&
$(1-p_{\VRF})p_{\CH}$&
$\Bc-p_{\CH}(\Bc+\Cw)$&
$-p_{\VRF}\Cp$\\
\hline
\begin{tabular}{c}
(\gametwo,\modelone), (\gametwo,\modeltwo)
\end{tabular}&
$\frac{\Cv}{\SM+\Cw}$,
$\frac{\Ct}{\SW+\Cp}$&
$(1-p_{\VRF})\mathbf{P}_{\CH}$&
\begin{tabular}{l}
$(p_{\VRF}(1-p_{\CH}^n)+(1-p_{\VRF})(1-\mathbf{P}_{\CH}))\Bc$\\
$-p_{\VRF}\Cv-(1-p_{\VRF})\mathbf{P}_{\CH}\Cw$\\
$-(1-p_{\VRF}p_{\CH})n\SM$
\end{tabular}&
$\SW-\Ct$\\
\hline
\begin{tabular}{c}
(\gametwo,\modelthree)
\end{tabular}&
$\frac{\Cv+\SM}{\SM+\Cw}$,
$\frac{\Ct}{\SW+\Cp}$&
$(1-p_{\VRF})\mathbf{P}_{\CH}$&
\begin{tabular}{l}
$(p_{\VRF}(1-p_{\CH}^n)+(1-p_{\VRF})(1-\mathbf{P}_{\CH}))\Bc$\\
$-p_{\VRF}\Cv-(1-p_{\VRF})\mathbf{P}_{\CH}\Cw$\\
$-p_{\VRF}(1-p_{\CH})n\SM$
\end{tabular}&
$-p_{\VRF}\Cp$\\
\hline
\begin{tabular}{c}
(\gamethree,\modelone)
\end{tabular}&
$0$,
$\frac{|W_i|\SW+\Ct}{|W_i|(\Cp+2\SW)}<p_{\VRF}\leq 1$&
$0$&
$\Bc-p_{\VRF}\Cv-n\SM$&
$|W_i|\SW-\Ct$\\
\hline
\begin{tabular}{c}
(\gamethree,\modeltwo)
\end{tabular}&
$0$,
$\frac{\Ct}{|W_i|(\Cp+\SW)}<p_{\VRF}\leq 1$&
$0$&
$\Bc-p_{\VRF}\Cv-n\SM$&
$|W_i|\SW-\Ct$\\
\hline
\begin{tabular}{c}
(\gamethree,\modelthree)
\end{tabular}&
$0$,
$\frac{\Ct}{|W_i|(\Cp+\SW)}<p_{\VRF}\leq 1$&
$0$&
$\Bc-p_{\VRF}(\Cv+n\SM)$&
$p_{\VRF}|W_i|\SW-\Ct$\\
\hline
\end{tabular}
\end{scriptsize}
\caption{Contractor Scenario}
\label{table:contractor}
\end{table*}

\subsection{Contractor Scenario}

The second scenario considered is a company that buys computational
power from Internet users and sells it to computation-hungry
costumers. In this case the company pays the users an amount $S=\SW=\SM$
for using their computing capabilities, and charges the consumers
another amount $\Bc>\SM$ for the provided service. Since the users are
not volunteers in this scenario, we assume that computing a task is not free
for them (i.e., $\Ct>0$), and they must have incentives to participate
(i.e., $U_{W_i}>0, \forall W_i\in W$). As in the previous case, we assume that the master
verifies and has a cost for accepting a wrong value, such that
$\Cw>\Cv>0$. Again, under these assumptions, the equilibria for games
\gameone and \gametwo collapse to unique equilibria and game~\gamefour can not be used. The different cases are summarized in Table~\ref{table:contractor}. Observe
that there are cases in this table in which the group has negative
expected utility $U_{W_i}$. Given that in this \afa{scenario} workers are not volunteers,
they will not accept to participate in such a game. This fact
immediately rules out games (\gameone,\modelthree) and (\gametwo,\modelthree)
and requires that $\SW>\Ct$ in general.  Similarly, this
restriction forces the master to use a value of $p_{\VRF} > \Ct/|W_i|\SW, \forall W_i\in W$ in
game (\gamethree,\modelthree). Finally, comparing games (\gamethree,\modelone) and (\gamethree,\modeltwo), it can be seen
that the master would never choose the former, because the lower bound
of $p_{\VRF}$ is smaller in the latter while the rest of expressions are
the same, which leads to a larger master utility. 

In this scenario, beyond choosing the game and number of workers $n$ 
as in the previous one, we assume that the master can also choose the reward
$\SW$ to the workers for correctly computing the task, and the
punishment $\Cp$ if they are caught returning an incorrect
value. All possible combined variations of these parameters yield a
huge number of cases to be considered. 
In what follows, we assume that the
master only can choose one of these parameters, while the rest are
predefined. A study of richer combinations is left for future work.

The following notation is used for clarity. Whenever a parameter
may be different among different games being compared, a
super-index indicates the game to which the parameter belongs. For
instance, $U_M^{(i,j)}$ is the utility of the master for game
$(i,j)$. $\SM$ and $\SW$ are referred to as simply $S$ \mma{($=\SM=\SW$).}

A simple observation of games (\gamethree,\modeltwo) and (\gamethree,\modelthree) leads to find that in
both cases it is convenient for the master to choose the smallest
possible value of $p_{\VRF}$. For this reason, in the following we assume
in these games values $p_{\VRF}^{(\gamethrees,\modeltwo)}=\frac{\Ct}{\Cp+S}+\gamma^{(\gamethrees,\modeltwo)}$ and
$p_{\VRF}^{(\gamethrees,\modelthree)}=\frac{\Ct}{S}+\gamma^{(\gamethrees,\modelthree)}$, for arbitrarily small
$\gamma^{(\gamethrees,\modeltwo)}>0$ and $\gamma^{(\gamethrees,\modelthree)}>0$~\footnote{We assume here the worst case scenario where $\min_{W_i\in W}\{|W_i|\}=1$. If a better lower bound can be guaranteed, a similar analysis taking it into account follows.}.


\subsubsection{Tunable $n$}
 Regarding games (\gameone,\modelone) and (\gametwo,\modelone),
 in this case the master has no control over $p_{\CH}$ or $p_{\VRF}$, since
 they are completely defined by the application parameters. Hence,
 the probability of accepting a wrong answer might be arbitrarily
 close to $1$, even for game (\gametwo,\modelone), because $\mathbf{P}_{\CH}$ grows with
 $n$ if $p_{\CH}>1/2$ as shown in Claim~\ref{claim:pL}. Given that we
 want to design a mechanism that can be applied to any setting, we
 rule out these games for this case.
In the case that $n$ is tunable, the benefit of the master in games
 (\gamethree,\modeltwo) and (\gamethree,\modelthree) decreases as $n$ increases. Hence for these games   the master chooses $n=1$. (So, $|W|=|W_i|=1$.)  Additionally, these games provide
 $\mathbf{P}_{wrong}=0$. 
 {\em Out of these games, (\gamethree,\modeltwo) is better iff 
 $\Cv>S(S/\Ct-1)(S/\Cp+1)$.}  
We highlight these observations in the following theorem.

\begin{theorem}
For any given set of payoff parameters, such that it can be characterized as the contractor scenario and where $S=\SM=\SW>\Ct$,
if the master gets to choose the number of workers,
in order to obtain the correct answer (with probability $1$) while maximizing the utility of the master, 
it is enough to assign the task to only one worker excluding the master from the game and

(i) if $\Cv<S(S/\Ct-1)(S/\Cp+1)$,
reward the worker only when the result is verified to be correct,
and verify with probability $p_{\VRF} = \varepsilon + \Ct/S$,

(ii) otherwise, 
reward the worker also if the result is not verified,
and verify with probability $p_{\VRF} = \varepsilon + \Ct/(S+\Cp)$,

for any positive $\varepsilon$ arbitrarily close to $0$.
\end{theorem}


\subsubsection{Tunable $\Cp$} 
We first compare games (\gamethree,\modeltwo) and (\gamethree,\modelthree),
\begin{align*}
U_M^{(\gamethrees,\modeltwo)}
&=\Bc-p_{\VRF}^{(\gamethrees,\modeltwo)}\Cv-nS\\
&=\Bc-\Ct\Cv/(S+\Cp^{(\gamethrees,\modeltwo)})-nS-\gamma^{(\gamethrees,\modeltwo)}\Cv
\end{align*}
and
\begin{align*}
U_M^{(\gamethrees,\modelthree)}
&=\Bc-p_{\VRF}^{(\gamethrees,\modelthree)}\Cv-p_{\VRF}^{(\gamethrees,\modelthree)}nS\\
&=\Bc-\Ct\Cv/S-n\Ct-\gamma^{(\gamethrees,\modelthree)}\Cv-\gamma^{(\gamethrees,\modelthree)}nS.
\end{align*}
{\em Thus, game (\gamethree,\modelthree) is better iff $nS(S/\Ct-1) > \Cv$ for small enough
$\gamma^{(\gamethrees,\modelthree)}$. Otherwise, (\gamethree,\modeltwo) is better for small enough
$\gamma^{(\gamethrees,\modeltwo)}$ and large enough $\Cp^{(\gamethrees,\modeltwo)}$.}
As argued in the previous case, in this case the master has no control
over $p_{\CH}$. Although the master can reduce $\Cp$ to increase $p_{\VRF}$, it
can not make $p_{\VRF}$ arbitrarily close to 1 to reduce
$\mathbf{P}_{wrong}$ in case $p_{\CH}$ is big (and consequently
$\mathbf{P}_{\CH}$). Then, some cases might lead to a big probability of
accepting the wrong answer. Thus, games (\gameone,\modelone) and (\gametwo,\modelone)
are ruled out from consideration.
We highlight these observations in the following theorem.

\begin{theorem}
For any given sets of workers and payoff parameters, except for $\Cp$ that is chosen by the master.
If the set of payoffs is such that $S=\SM=\SW>\Ct$ and it can be characterized as the contractor scenario,
in order to obtain the correct answer (with probability $1$) while maximizing the utility of the master, 
it is enough to exclude the master from the game and

(i) if $\Cv<nS(S/\Ct-1)$,
reward the workers only when the result is verified to be correct,
and verify with probability $p_{\VRF} = \varepsilon + \Ct/S$,

(ii) otherwise, 
reward the workers also if the result is not verified,
set $\Cp$ so that $\Cv<nS(S/\Ct-1)(S/\Cp+1)$,
and verify with probability $p_{\VRF} = \varepsilon + \Ct/(S+\Cp)$,

for any positive $\varepsilon$ arbitrarily close to $0$.
\end{theorem}


\subsubsection{Tunable $S=\SW=\SM$ within the interval $(\Ct,\Bc)$} 
In this case
 $n$ is fixed, and given that we do not make any assumptions about its
 magnitude, we evaluate game~\gameone while evaluating game~\gametwo for an
 arbitrary $n$.
 Using calculus, the utility of the master for game (\gamethree,\modeltwo) is
 maximum when $S_{\max}^{(\gamethrees,\modeltwo)}=\pm\sqrt{\Cv\Ct/n}-\Cp$. Due to the
 aforementioned constraints, only values in the interval $(\Ct,\Bc)$
 are valid for $S$. Assuming then that $\Ct<S_{\max}^{(\gamethrees,\modeltwo)}<\Bc$,
 the utilities are
 
$U_M^{(\gamethrees,\modeltwo)}(S=S_{\max}^{(\gamethrees,\modeltwo)})= \Bc-2\sqrt{n\Cv\Ct}+n\Cp$ and

$U_M^{(\gamethrees,\modelthree)}= \Bc-\Ct\Cv/S^{(\gamethrees,\modelthree)}-n\Ct-\gamma^{(\gamethrees,\modelthree)}(\Cv+nS^{(\gamethrees,\modelthree)})$.

{\em Since $U_M^{(\gametwos,\modelone)}\leq \Bc$, game (\gamethree,\modeltwo) is better than game (\gametwo,\modelone)
whenever $n > 4\Cv\Ct/\Cp^2$.  On the other hand, game (\gamethree,\modelthree) is better
than game (\gamethree,\modeltwo) if $S^{(\gamethrees,\modelthree)}> \Ct\Cv/(2\sqrt{n\Cv\Ct}-n(\Cp+\Ct))$, for
small enough $\gamma^{(\gamethrees,\modelthree)}$.}
We highlight these observations in the following theorem.

\begin{theorem}
For any given sets of workers and payoff parameters, except for $S=\SM=\SW$ that is chosen by the master.
If the set of payoffs can be characterized as the contractor scenario,
in order to obtain the correct answer (with probability $1$) while maximizing the utility of the master, 
it is enough to exclude the master from the game,
reward the workers only when the result is verified to be correct,
set $S > \Ct\Cv/(2\sqrt{n\Cv\Ct}-n(\Cp+\Ct))$,
and verify with probability $p_{\VRF} = \varepsilon + \Ct/S$,
for any positive $\varepsilon$ arbitrarily close to~$0$.
\end{theorem}

In order to show a scenario where game (\gametwo.\modelone) is better, we assume now
that $\Cw\geq 2\Cv$. Then, under this assumption, $p_{\CH}\leq 1/2$. The
following claim that makes use of this fact will be useful.

\begin{claim}
\label{claim:pL}
For game~\gametwo, \afa{ let $\mathbf{P}_{\CH}(n)$ denote} the probability that the majority
out of $n$ workers \afa{cheat. If} the probability that a worker
\afa{cheats is} $p_{\CH}\leq\frac{1}{2}$, then
$\mathbf{P}_{\CH}(n+2)\leq\mathbf{P}_{\CH}(n)$.
\end{claim}
\begin{proof}
 Let $\mathbf{P}_{\CH}(n,>1)$ be the probability that, out of $n$
 workers, the number of cheaters exceed the number of honest workers
 by more than one (i.e., at least 3 given that we consider only odd
 number of workers), $\mathbf{P}_{\CH}(n,=1)$ by exactly one, and
 $\mathbf{P}_{\overline{{\CH}}}(n,=1)$ be the probability that the number
 of honest workers exceed the number of cheaters by exactly one.
Then,
$\mathbf{P}_{\CH}(n+2)=
\mathbf{P}_{\CH}(n,>1)(p_{\CH}^2+(1-p_{\CH})^2)
+\mathbf{P}_{\CH}(n,=1)(p_{\CH}^2+2p_{\CH}(1-p_{\CH}))
+\mathbf{P}_{\overline{{\CH}}}(n,=1)p_{\CH}^2$.
Bounding $p_{\CH}$ the claim follows. 
\end{proof}

From the previous claim, given that $\mathbf{P}_{\CH}=1/2$ for $p_{\CH}=1/2$,
we conclude that $\mathbf{P}_{\CH}\leq 1/2$. Using that $p_{\CH}\leq 1/2$,
$\mathbf{P}_{\CH}\leq 1/2$, and $\Cw>2\Cv$, the utility of the master for
game (\gametwo,\modelone) is

\begin{align*}
U_M^{(\gametwos,\modelone)} \geq &~\frac{1}{2}\Bc - p_{\VRF}^{(\gametwos,\modelone)}\Cv 
- \frac{1}{2}(1-p_{\VRF}^{(\gametwos,\modelone)})\Cw -nS^{(\gametwos,\modelone)}\\
= &~\frac{1}{2}\Bc - p_{\VRF}^{(\gametwos,\modelone)}\Cv -
\frac{1}{2}\Cw
+ \frac{1}{2}p_{\VRF}^{(\gametwos,\modelone)}\Cw -nS^{(\gametwos,\modelone)}\\
\geq &~\frac{1}{2}(\Bc - \Cw) -nS^{(\gametwos,\modelone)}.
\end{align*}

{\em As shown before, game (\gamethree,\modeltwo) is better than game (\gamethree,\modelthree) when
$$\Bc<\Ct\Cv/(2\sqrt{n\Cv\Ct}-n(\Cp+\Ct)).$$}  Comparing games (\gametwo,\modelone) and
(\gamethree,\modeltwo) when $\Ct<\sqrt{\Cv\Ct/n}-\Cp<\Bc$,
we have $(\Bc - \Cw)/2 -nS^{(\gametwos,\modelone)}\geq
\Bc-2\sqrt{n\Cv\Ct}+n\Cp$. {\em Therefore, game (\gametwo,\modelone) is better whenever}
\begin{align}
\Ct\leq S^{(\gametwos,\modelone)}\leq& 
2\sqrt{\frac{\Cv\Ct}{n}} 
-\frac{1}{2n}(\Bc +\Cw)-\Cp \label{ineq:S}
\end{align}
All three conditions are feasible simultaneously for big enough $\Cv$,
therefore there exists a scenario for which game (\gametwo,\modelone) is better.
Notice that under the aforementioned condition, for game (\gamethree,\modeltwo) to be
better, i.e., $n> 4\Cv\Ct/\Cp^2$, it must be true that $\Cp> 2\sqrt{\Cv\Ct/n}$ and the
inequality~(\ref{ineq:S}) does not hold.


\subsection{Computational Issues}
\label{section:complex}
In previous sections, a mechanism for the master to choose games, payoff models, and appropriate values of $p_\VRF$ for different scenarios was designed (based on Algorithm~\ref{alg1}). A natural question is what is the computational cost of using such mechanism. In addition to simple arithmetical calculations, there are two kinds of relevant computations required: binomial probabilities and verification of conditions for Nash equilibria. Both computations are $n$-th degree polynomial evaluations and can be carried out using any of the well-known numerical tools~\cite{horner} with polynomial asymptotic cost. These numerical methods yield only approximations, but all these calculations are performed either to decide in which case the parameters fit in, or to assign a value to $p_\VRF$, or to compare utilities. Given that these evaluations and assignments were obtained in the design as inequalities or restricted only to lower bounds, it is enough to choose the appropriate side of the approximation in each case.
Regarding the computational resources that the workers require to carry out these calculations, notice that the choice of $p_\VRF$ in the mechanism only yields a unique NE. Then, in order to make the computation feasible to the workers, the master sends together with the task a certificate proving such equilibrium. Such a certificate is the value of $p_\VRF$, payoff values, game, and payoff model, which is enough to verify uniqueness.

\section{Conclusions}
\label{sec:Conc}

In this paper we consider computational systems in which a master processor assigns tasks for execution to rational workers. We have defined the general model and cost-parameters, and we have proposed and analyzed several games that the master can choose to play in order to achieve high reliability at low cost. Based on our game analysis, we have designed appropriate algorithmic mechanisms for two realistic scenarios of these kinds of systems. 

\textcolor{black}{
While volunteer computing systems used in practice, like BOINC, use redundant task allocation to detect erroneous answers~\cite{boinc}, it is known that this technique can not guarantee correctness in presence of collusion~\cite{CCS}.
Interestingly, our results show that verifying with very small probability can be used to prevent erroneous answers from selfish workers, even under collusion.
A richer exploration of the practical implications of the results in this paper is left for future work. 
}

\textcolor{black}{In order to expand and generalize our model,} we plan to design more complex mechanisms where more than one parameter at a time is tunable by the master, and consider other realistic scenarios where our work can be applied.
It would also be interesting to consider the case where the workers and/or the master do not
have complete information of all the system parameters \textcolor{black}{(that is, consider Bayesian Mechanism Design, see e.g.~\cite{HR_STOC08})}. \textcolor{black}{Furthermore, we plan to consider
the more general problem in which there is a sequence of tasks whose values must be
reliably obtained. To this respect, cheater detection mechanisms, as the one considered
in~\cite{KSW_IPDPS08}, must be deployed and multiple-rounds protocols must be designed.}  



\end{document}